\documentclass[
 reprint,
 amsmath,amssymb,
 aps,
longbibliography
]{revtex4-1}

\usepackage{graphicx}
\usepackage{dcolumn}
\usepackage{bm}

\usepackage{color}
\usepackage{amsthm}
\usepackage{tikz}
\usepackage{amsfonts}
\usepackage{amssymb}
\usepackage[mathscr]{eucal}
\usepackage[normalem]{ulem}
\usepackage[colorlinks]{hyperref} 
\usepackage{braket}

\usepackage{cleveref}

\newtheorem{conjecture}{Conjecture}

\newtheorem{thm}{Theorem}

\newcommand{\N}{{\sf N}}
\newcommand{\D}{\sf{D}}

\newcommand{\I}{\mathscr{I}}

\newcommand{\uI}{\underline{\mathscr{I}}}

\newcommand{\uK}{\underline{\mathscr{K}}}

\newcommand{\ent}{{\sf S}}  %
\newcommand{\entR}{{\sf R}} 
\newcommand{\mi}{{\sf I}}  
\newcommand{\face}{\mathscr{F}}   

\definecolor{THcolor}{rgb}{0.1,0.7,0.1}

\definecolor{VHcolor}{rgb}{0.7,0.3,0.9}

\definecolor{MRocolor}{rgb}{0.1,0,1}

\begin{document}

\preprint{APS/123-QED}

\hfill CALT-TH 2023-027

\title{A gap between holographic and quantum mechanical extreme rays of the subadditivity cone}

\author{Temple He}
\affiliation{Walter Burke Institute for Theoretical Physics \\ California Institute of Technology, Pasadena, CA 91125 USA}
\email{templehe@caltech.edu}

\author{Veronika E. Hubeny}
\affiliation{Center for Quantum Mathematics and Physics (QMAP)\\ 
Department of Physics \& Astronomy, University of California, Davis, CA 95616 USA}
\email{veronika@physics.ucdavis.edu}

\author{Massimiliano Rota}
\affiliation{Center for Quantum Mathematics and Physics (QMAP)\\ 
Department of Physics \& Astronomy, University of California, Davis, CA 95616 USA}
\affiliation{Institute for Theoretical Physics, University of Amsterdam,
Science Park 904, Postbus 94485, 1090 GL Amsterdam, The Netherlands}
\email{maxrota@ucdavis.edu}

\date{\today}

\begin{abstract}
We show via explicit construction that for six or more parties, there exist extreme rays of the subadditivity cone that can be realized by quantum states, but not by holographic states. This is a counterexample to a conjecture first formulated in \cite{Hernandez-Cuenca:2022pst}, and implies the existence of deep holographic constraints that restrict the allowed patterns of independence among various subsystems beyond the universal quantum mechanical restrictions.   
\end{abstract}


\maketitle


\section{\label{sec:intro} Introduction}

A central question in the context of the gauge/gravity duality \cite{Maldacena:1997re,Gubser:1998bc,Witten:1998qj} is to understand how the bulk classical geometry is encoded in the entanglement structure of the boundary state, and one might hope to extract useful information about such encoding by investigating properties of the von Neumann entropy which are specific to this setting. The discovery of \textit{monogamy of mutual information} (MMI) \cite{Hayden:2011ag,Wall:2012uf} showed that for \textit{geometric states}, i.e., states of holographic CFTs which are dual to classical geometries, the HRRT prescription \cite{Ryu:2006bv,Hubeny:2007xt} implies that the entropies of spatial subsystems in the boundary CFT satisfy constraints that in general do not hold for arbitrary quantum systems. Since then, new holographic entropy inequalities have been found, and the \textit{holographic entropy cone} (HEC) \cite{Bao:2015bfa} has been studied extensively \cite{Marolf:2017shp,Rota:2017ubr,Cui:2018dyq,Hubeny:2018trv,Bao:2018wwd,Hubeny:2018ijt,Cuenca:2019uzx,Czech:2019lps,He:2019ttu,He:2020xuo,Avis:2021xnz,Czech:2022fzb}.

As the number of parties $\N$ increases, the search for new inequalities quickly becomes computationally unfeasible due to the fact that the combinatorics governing the number of inequalities typically grows doubly exponentially as a function of $\N$. Furthermore, fixing $\N$ is immaterial in QFT, since one can always imagine further partitioning the $\N$ regions into smaller subregions. For these reasons, \cite{Hernandez-Cuenca:2022pst} took a different approach to the characterization of the HEC. Rather than looking for the explicit expression of the inequalities at some given $\N$, \cite{Hernandez-Cuenca:2022pst} attempted to provide a more implicit description of the HEC for an arbitrary number of parties by relating it to the \textit{quantum entropy cone} (QEC) \cite{1193790}, and to distill the essential information that would allow for its reconstruction (at least in principle). Drawing from the ideas of \cite{Hubeny:2018trv,Hubeny:2018ijt}, \cite{Hernandez-Cuenca:2022pst} suggested that this essential information is the solution to the \textit{holographic marginal independence problem} (HMIP) \cite{Hernandez-Cuenca:2019jpv}.  

The HMIP is the restriction of the more general \textit{quantum marginal independence problem} (QMIP), introduced in \cite{Hernandez-Cuenca:2019jpv}, to geometric states. The QMIP asks the following question: Given an $\N$-party system and a complete specification of the presence of correlation (or conversely the lack thereof) among the various subsystems, is there a density matrix that satisfies these constraints? This problem can be conveniently formalized using the polyhedral cone in entropy space carved out by all instances of \textit{subadditivity} (SA) at given $\N$, called the \textit{subadditivity cone} (SAC). The SAC is an outer bound to the HEC, and a \textit{pattern of marginal independence} (PMI) is defined as the linear subspace spanned by one of its faces \footnote{This is the definition of a PMI given in \cite{Hernandez-Cuenca:2022pst}, whereas the original definition from \cite{Hernandez-Cuenca:2019jpv} was weaker.}. The HMIP then asks which faces of the SAC, and therefore which PMIs, can be reached by the HEC.

The analysis of \cite{Hernandez-Cuenca:2022pst} suggested that the HEC can be reconstructed if the solution to the \textit{extremal} version of the HMIP is known, i.e., if it is known which \textit{extreme rays} of the SAC (or equivalently 1-dimensional PMIs) can be realized by geometric states. Specifically, \cite{Hernandez-Cuenca:2022pst} provided strong evidence that the extreme rays of the HEC can simply be obtained from 1-dimensional PMIs involving more subsystems by particular projections that correspond to coarse-grainings \footnote{Notice that this implies that to reconstruct the extreme rays of the $\N$-party HEC, one needs to know the extreme rays of the SAC for some $\N'\geq\N$. We will return to this point in the next section.}. Holographic entropy inequalities can then be derived, at least in principle, using standard algorithms to convert the description of a polyhedral cone in terms of extreme rays into one given by facets \footnote{Since no efficient algorithm is known, such explicit derivation of the inequalities remains computationally unfeasible for large $\N$.}. 

If this reconstruction is indeed possible, the characterization of the HEC would reduce to the characterization of the set of extreme rays of the SAC that can be realized by geometric states. A natural guess suggested in \cite{Hernandez-Cuenca:2022pst}, which is consistent with the SAC for all $\N\leq 5$, is that these are all the extreme rays that can be realized in quantum mechanics. The aim of this letter is to show that this is \textit{not} the case, implying that even if the reconstruction argued in \cite{Hernandez-Cuenca:2022pst} is possible, extremal quantum marginal independence alone is not enough to fully characterize the HEC, since there exist deeper constraints which restrict the set of extremal PMIs that can be realized in holography. 

The structure of the paper is as follows. In \S\ref{sec:HECfromHMIP} we review the definition of the HEC from \cite{Bao:2015bfa}, the notion of holographic marginal independence from \cite{Hernandez-Cuenca:2019jpv}, and the reconstruction of the HEC from the solution to the extremal HMIP argued in \cite{Hernandez-Cuenca:2022pst}. In \S\ref{sec:lattice_techniques} we review the machinery from \cite{He:2022bmi}, which allows us to efficiently derive extreme rays of the SAC that satisfy \textit{strong subadditivity} (SSA) and therefore have a chance of being realizable by quantum states. We will use these techniques to find such an extreme ray that violates MMI, and therefore cannot be realized by any geometric state. Nevertheless, we will demonstrate in \S\ref{sec:example} that this extreme ray is in fact realized by a quantum state. Finally, in \S\ref{sec:discussion} we comment on the implications of this result for the characterization of the HEC for an arbitrary number of parties.

\section{The holographic entropy cone from marginal independence} \label{sec:HECfromHMIP}

\subsection{\label{subsec:def}Definition of the holographic entropy cone}

The HEC was introduced in \cite{Bao:2015bfa}, following the analogous program for arbitrary quantum states \cite{1193790}, as a convenient framework to analyze the set of inequalities implied by the Ryu–Takayanagi (RT) formula \cite{Ryu:2006bv}. For our purposes, it will be sufficient to view the HEC as the convex cone of entropy vectors which are realized by \textit{graph models of holographic entanglement}, which we now briefly review. For the original definition of the HEC, and for a detailed explanation of how a graph model is related to the RT surfaces that compute the entropies of a collection of boundary spatial subsystems, we refer the reader to \cite{Bao:2015bfa} (see also \cite{Marolf:2017shp,Hubeny:2018trv,Hubeny:2018ijt} for additional subtleties).

An $\N$-party graph model is a simple weighted graph $G=(V,E)$ with positive weights, a specification of a subset $\partial V\subseteq V$ of vertices called \textit{boundary vertices}, and a surjective (but not necessarily injective) map $\xi:\partial V\rightarrow [\N+1]=\{1,2,\ldots,\N+1\}$, where $1,\ldots,\N$ label the parties and $\N+1$ labels the purifier. Given a graph model, one associates to it an \textit{entropy vector} $\vec\ent\in\mathbb{R}^{\D}$, with $\D=2^\N-1$, as follows. For a non-empty subset $\I\subseteq [\N]$, a cut ``homologous to $\I$'' ($\I$-cut) is a subset $V_{\I}\subset V$ such that $\partial V\cap V_{\I}=\xi^{-1}(\I)$, where $\xi^{-1}$ denotes the pre-image of $\xi$. The \textit{cost} of any such cut is the sum of the weights of the edges that connect a vertex in $V_{\I}$ to one in $V_{\I}^c$, the complement of $V_{\I}$ in $V$. The entropy $\ent_{\I}$ is then defined as the cost of the $\I$-cut with \textit{minimal cost}.

It is straightforward to see that the set of entropy vectors obtained from all such graph models at any fixed $\N$ is a convex cone. This follows from the fact that given any two graph models $G_1,G_2$ with entropy vectors $\vec \ent_1,\vec \ent_2$, the conical combination $\vec\ent=\alpha\vec\ent_1+\beta\vec\ent_2$, with $\alpha,\beta>0$, is realized by the graph $G=\alpha G_1\oplus \beta G_2$, where $\alpha G_1$ is a graph model obtained from $G_1$ by rescaling the weights with the coefficient $\alpha$ (and similarly for $\beta G_2$), and $\oplus$ denotes the disjoint union. As shown in \cite{Bao:2015bfa}, this cone is identical to the HEC, and it is polyhedral for any $\N$, i.e., it is specified by a finite number of inequalities, or equivalently, by a finite number of extreme rays. It was further shown in \cite{hayden2016holographic} that, for any $\N$, the cone of graph models, or equivalently the HEC, is contained in the QEC \footnote{This is not a priori obvious because the original definition of the HEC in \cite{Bao:2015bfa} is a purely geometric one, and it does not assume that the bulk geometry corresponds to a CFT state \cite{Marolf:2017shp}.}.

\subsection{\label{subsec:ehimp}Formalization of the HMIP and its extreme version}

An obvious outer bound to the QEC, and therefore also the HEC, is the SAC, i.e., the polyhedral cone carved out by all instances of SA:
\begin{equation}
\label{eq:mi_instances}
    \ent_{\uI}+\ent_{\uK}-\ent_{\uI\uK}\geq 0\quad \forall\, \uI,\uK \subset [\N+1],\quad \uI\cap\uK=\varnothing ,
\end{equation}
where underlined indices indicate subsystems that can include the purifier $\N+1$ (as opposed to $\I \subseteq [\N]$). Thus, the Araki-Lieb inequality conveniently takes the form of an SA involving the purifier. 

Consider now a face $\face$ of the SAC and a vector $\vec\ent\in\text{int}(\face)$. Notice that the collection of SA instances which are saturated by $\vec\ent$ is independent from the specific choice of this vector. Furthermore, the saturation of SA is equivalent to the vanishing of the \textit{mutual information} (MI) $\mi(\uI:\uK)$, which in quantum mechanics is attained if and only if the density matrix factorizes, i.e., if the subsystems $\uI$ and $\uK$ are independent. We can then interpret the linear subspace spanned by $\face$ as corresponding to a specification of which subsystems are independent and which manifest some correlation, while remaining agnostic about the specific values of the entropies. This subspace is in fact a PMI as defined in \cite{Hernandez-Cuenca:2022pst}.

The QMIP and HMIP then ask which PMIs can be realized by quantum states and graph models, respectively. More specifically, they ask which PMIs correspond to faces of the SAC such that there exist in the interior at least one entropy vector realized by a quantum state or, respectively, a graph model. In the following sections we will be interested in the extremal version of these problems, which we denote by EQMIP and EHMIP, and only focus on the 1-dimensional PMIs. These correspond to the 1-dimensional faces of the SAC, i.e., its extreme rays.

\subsection{\label{subsec:reconstruction}Reconstructing the HEC from the solution to the EHMIP}

The intuition that the solution to the HMIP could provide sufficient information for the derivation of the HEC was first suggested in \cite{Hubeny:2018trv,Hubeny:2018ijt}, based on the observation that the PMI of a choice of boundary regions in a geometric state is captured by the connectivity of the RT surfaces, while the specific value of the entropy is immaterial in QFT because it depends on the choice of cut-offs. This intuition was then further developed in \cite{Hernandez-Cuenca:2022pst}, which formulated the following conjecture and checked that it holds for $\N\leq 5$ \footnote{Most of the machinery developed in \cite{Hernandez-Cuenca:2022pst} was used to show that \Cref{conj:1} would follow from other purely graph theoretic conjectures, but this machinery is not necessary for the purpose of this letter.}:

\begin{conjecture}
\label{conj:1}
    For any extreme ray $\vec\entR_{\N}$ of the $\N$-party \emph{HEC}, there exists for some $\N' \geq \N$ an extreme ray $\vec \entR'_{\N'}$ of the $\N'$-party \emph{SAC} such that $\vec \entR'_{\N'}$ can be realized by a graph model and 
    \begin{equation}
        \Lambda_{\N'\rightarrow \N}\vec \entR'_{\N'}=\entR_{\N} ,
    \end{equation}
    where $\Lambda_{\N'\rightarrow \N}$ is a map associated to a coarse-graining of the $\N'$ parties into $\N$ blocks \emph{\footnote{The precise expression of $\Lambda_{\N'\rightarrow \N}$ is not necessary for this discussion, but for clarity we give a simple example of such a projection. Consider the 3-party entropy vector $\vec\ent=(1,1,1;2,2,2;1)$ for $A,B,C$ (the entropies are ordered conventionally as exemplified in \eqref{eq:lex} for $\N=6$), and define the new parties $A'=A,\,B'=BC$. The 2-party entropy vector for these coarse-grained parties is then $\vec\ent'=(1,2;1)$, and it is obtained from $\vec\ent$ by simply dropping the components that separate $B$ from $C$, i.e., $B,C,AB,AC$.}}.
\end{conjecture}

If proven to be true, \Cref{conj:1} would have important implications for the characterization of the HEC. It would imply that for any $\N$, there exists some finite $\N'\geq \N$ such that the $\N$-party HEC can be obtained as the conical hull of all possible coarse-grainings of the extreme ray of the $\N'$-party SAC realizable by graph models \cite{Hernandez-Cuenca:2022pst}. In other words, to reconstruct the $\N$-party HEC it would be sufficient to know the solution to the EHMIP for a certain $\N'\geq\N$ which depends on $\N$ (it was shown in \cite{Hernandez-Cuenca:2022pst} that for $\N=3,4$ it suffices to have $\N'=\N$, whereas for $\N=5$ one needs $\N'=8$).

While a proof of \Cref{conj:1} is still lacking, and it is far beyond the scope of this letter, considering the strong evidence given in \cite{Hernandez-Cuenca:2022pst}, it is natural to focus on the solution to the EHMIP, since this distills the essential information underlying the HEC. The immediate question then is whether there is some physical principle that identifies the SAC extreme rays realizable by graph models. A possibility suggested in \cite{Hernandez-Cuenca:2022pst}, which holds for $\N \leq 5$, is that there is in fact nothing special about graph models, and that the EHMIP and EQMIP have the same solution:

\begin{conjecture}
\label{conj:2}
    For any $\N$, all extreme rays of the $\N$-party \emph{SAC} that can be realized by quantum states can also be realized by graph models.
\end{conjecture}

In the rest of this work we will construct a counterexample to \Cref{conj:2}. Its implication for the characterization of the HEC will then be discussed in \S\ref{sec:discussion}.

\section{Deriving extreme rays of the SAC which satisfy SSA}
\label{sec:lattice_techniques}

As we mentioned, \Cref{conj:2} holds for $\N\leq 5$, so to look for a counterexample we need to consider at least $\N=6$. However, in this case the combinatorics of the faces of the SAC is sufficiently complicated that an explicit derivation of all extreme rays is not feasible even using state of the art algorithms \cite{Normaliz:301}. Instead, we should restrict ourselves to the relatively few extreme rays which satisfy SSA, since these are the only ones that can possibly be realized by quantum states. 

A first useful result in this direction was recently obtained in \cite{He:2022bmi}.

\begin{thm}
\label{thm:bp_theorem}
    For any $\N$, all extreme rays of the \emph{SAC} that can possibly be realized by quantum states (other than the ones realized by Bell pairs) belong to the face that spans the subspace given by
    \begin{equation}
    \mi(\ell:\ell')=0\quad \forall\ell,\ell'\in [\N+1] .
    \end{equation}
\end{thm}
\begin{proof}
    See Corollary 1 in \cite{He:2022bmi}.
\end{proof}

While \Cref{thm:bp_theorem} gives a considerable speed-up in the computation of the extreme rays of the SAC that satisfy SSA \footnote{The computation of all the extreme rays of the SAC for $\N=5$ takes several days on a standard laptop, but the extreme rays of the face identified by \Cref{thm:bp_theorem} only takes a few minutes.}, it is still not sufficient to obtain all such extreme rays for $\N=6$. To obtain these rays, one can generalize \Cref{thm:bp_theorem} and derive new constraints using extreme rays that are already known, in a similar fashion to how \cite{He:2022bmi} proved \Cref{thm:bp_theorem} using the extreme rays realized by Bell pairs. This program is currently being explored systematically in \cite{He:2022wip}, but for the purpose of providing a counterexample to \Cref{conj:2}, it suffices to construct a single extreme ray with the requisite conditions.

Labeling the six parties by $A,B,C,D,E,F$, and ordering the components of an entropy vector lexicographically, as in
\begin{equation}
\label{eq:lex}
    (A,\ldots,F;AB,AC,\ldots,EF;ABC,\ldots;ABCDEF),
\end{equation}
the example we consider is
\begin{align}
\begin{split}
\label{eq:ray}
   \entR_6 = (&2, 1, 1, 1, 2, 2;\, 3, 3, 3, 4, 4, 2, 2, 3, 3, 2, 3, 3, 3, 3, 4; \\
    & 2, 4, 5, 5, 4, 5, 5, 3, 5, 4, 3, 4, 4, 4, 4, 5, 4, 4, 5, 3;\, 3, \\
    & 4, 4, 4, 4, 3, 4, 4, 3, 3, 3, 5, 4, 4, 4;\, 3, 3, 2, 2, 2, 3;\, 1).
\end{split}
\end{align}
The reader can easily verify that \eqref{eq:ray} is indeed an extreme ray of the 6-party SAC by first checking that it satisfies all instances of SA, and that the set of vanishing MI instances has rank ${\D}-1=62$, where ${\D}=2^6-1$ is the dimension of the $\N=6$ entropy space. Furthermore, as the reader can check, \eqref{eq:ray} violates one instance of MMI, in particular
\begin{equation}
\label{eq:mmi_violation}
    -\mi_3\,(A:BC:DE) = -2 \not\geq 0,
\end{equation}
where
\begin{align}
    & \mi_3\,(X:Y:Z):= \nonumber\\
    &\qquad \ent_X+\ent_{Y}+\ent_{Z}-\ent_{XY}-\ent_{XZ}-\ent_{YZ}+\ent_{XYZ}.
\end{align}
This implies that \eqref{eq:ray} cannot be realized by a graph model.

To complete the proof that \Cref{conj:2} is false, we now need to show that even though \eqref{eq:ray} cannot be realized by a graph model, it is nevertheless the entropy vector of a quantum state. This is the goal of the next section.

\section{\label{sec:example} Quantum state realization}

To show that \eqref{eq:ray} is the entropy vector of a quantum state we will use the \textit{hypergraph models} introduced in \cite{Bao:2020zgx}. Hypergraph models are defined analogously to the graph models presented above, with the only difference being that in addition to edges one also allows for hyperedges connecting three or more vertices. Therefore we only need to clarify under what circumstances the weight of a hyperedge contributes to the cost of an $\I$-cut. As for standard edges, given an $\I$-cut $V_{\I}$, and a hyperedge $h$ (thought of as a collection of vertices), the weight of $h$ contributes to the cost of the cut if and only if $h$ contains at least one vertex in both $V_{\I}$ and $V_{\I}^c$. As usual, the entropy of $\I$ is then given by the cost of the $\I$-cut with minimal cost.

We can now try to construct a hypergraph model whose entropy vector is the extreme ray \eqref{eq:ray}. There is currently no systematic procedure to construct a hypergraph (or even graph) realization of a given entropy vector, but a convenient starting point is the observation from the previous section that \eqref{eq:ray} violates only a single instance of MMI (cf.~\eqref{eq:mmi_violation}). The prototypical example of a quantum state that violates MMI is the GHZ state, which is realized by a hypergraph with just a single hyperedge. To realize \eqref{eq:ray} we then start from a hypergraph with a single weight 2 hyperedge (2 is the value of the instance of $\mi_3$ in \eqref{eq:mmi_violation} obtained from \eqref{eq:ray}) connecting four vertices labeling the coarse-grained subsystems $A,BC,DE,FO$. With a few manipulations we then arrive at the hypergraph shown in \Cref{fig:hg}.

\begin{figure}[tb]
    \centering
    \begin{tikzpicture}
    \draw[fill=yellow,draw=none] (-0.2,0.4) .. controls (-1,2) and (1,2) .. (0.2,0.4) .. controls (0.1,0.2) and (0.2,0.1) .. (0.4,0.2) .. controls (2,1) and (2,-1) .. (0.4,-0.2) .. controls (0.2,-0.1) and (0.1,-0.2) .. (0.2,-0.4) .. controls (1,-2) and (-1,-2) .. (-0.2,-0.4) .. controls (-0.1,-0.2) and (-0.2,-0.1) .. (-0.4,-0.2) .. controls (-2,-1) and (-2,1) .. (-0.4,0.2) .. controls (-0.2,0.1) and (-0.1,0.2) .. (-0.2,0.4);
    \draw[very thick,blue]  (0,1) -- (-1,2) ;
    \draw[very thick,blue]  (0,1) -- (1,2);
    \draw[very thick,blue]  (1,0) -- (2,1);
    \draw[very thick,blue]  (1,0) -- (2.366,-0.366);
    \draw[very thick,blue]  (1,0) -- (1.366,-1.366);
    \draw[very thick,blue]  (0,-1) -- (1.366,-1.366);
    \draw[very thick,blue]  (0,-1) -- (-1.366,-1.366);
    \filldraw[red] (0,1) circle (2.5pt);
    \filldraw[red] (1,0) circle (2.5pt);
    \filldraw[red] (0,-1) circle (2.5pt);
    \filldraw (-1,2) circle (2.5pt); 
    \filldraw (1,2) circle (2.5pt); 
    \filldraw (-1,0) circle (2.5pt); 
    \filldraw (2,1) circle (2.5pt); 
    \filldraw (2.366,-0.366) circle (2.5pt); 
    \filldraw (1.366,-1.366) circle (2.5pt); 
    \filldraw (-1.366,-1.366) circle (2.5pt); 
    \draw[orange] (0,0) node[]{{\footnotesize $2$}};
    \draw[blue] (1.7,0) node[right]{{\footnotesize $2$}};
    \draw[] (-1.4,0) node{{\footnotesize $A$}};
    \draw[] (-1.3,2.3) node{{\footnotesize $B$}};
    \draw[] (1.3,2.3) node{{\footnotesize $C$}};
    \draw[] (2.3,1.3) node{{\footnotesize $D$}};
    \draw[] (2.7,-0.45) node{{\footnotesize $E$}};
    \draw[] (1.666,-1.666) node{{\footnotesize $F$}};
    \draw[] (-1.75,-1.55) node{{\footnotesize $O$}};
    \end{tikzpicture}
    \caption{The hypergraph model that realizes the entropy vector \eqref{eq:ray}. The (yellow) blob describes the only hyperedge in the model, which connects the boundary vertex $A$ and the three bulk vertices (red), and has weight 2. All other edges (blue) have weight 1, except for the $E$ leaf which has weight 2.}
    \label{fig:hg}
\end{figure}
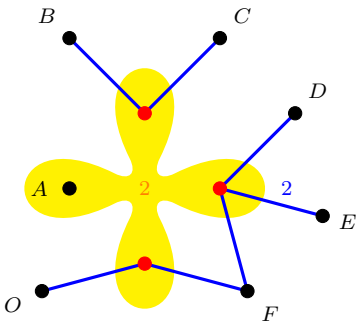

We are now ready to prove the main result of this letter.

\begin{thm}
    The extreme ray of the 6-party \emph{SAC} given in \eqref{eq:ray} is the entropy vector of a quantum state.
\end{thm}
\begin{proof}
    We leave it as a simple exercise for the reader to explicitly verify that the entropy vector of the hypergraph model shown in \Cref{fig:hg} is precisely \eqref{eq:ray}. The fact that \eqref{eq:ray} is the entropy vector of a quantum state then follows immediately from the result of \cite{Walter:2020zvt}, which showed that (similar to the case of a standard graph) any entropy vector realizable by a hypergraph model is the entropy vector of a quantum stabilizer state.
\end{proof}

\section{\label{sec:discussion} Discussion}

We conclude with a few comments about the implications of the failure of \Cref{conj:2} for the reconstruction and the physical interpretation of the HEC. If \Cref{conj:1} is true, the HEC can be fully reconstructed from the solution to the EHMIP. In that case, if \Cref{conj:2} were also true, the HEC would be the largest possible polyhedral cone compatible with this reconstruction procedure and quantum mechanics. One should then question the actual physical meaning of holographic entropy inequalities, since they would just follow from the bound on $\N$, which is artificial in QFT. The failure of \Cref{conj:2} proves that this is not the case, even if the reconstruction procedure of \Cref{conj:1} can indeed be achieved. In that case, one should then try to understand what distinguishes the extreme rays of the SAC that can be realized by graph models from the larger set of quantum mechanical ones. 

A first step in this direction is the derivation of all extreme rays which are compatible with SSA for $\N=6$ \cite{He:2022wip}. Since this is the smallest value of $\N$ where the solutions to the EQMIP and EHMIP differ, it is a useful testing ground to develop intuition. Next, one should determine which of these rays can be realized in quantum mechanics, and as demonstrated here, the hypergraph construction could be a useful tool in this direction. Ultimately, one may then attempt to use the construction suggested in \cite{Bao:2020zgx}, or the techniques from \cite{hayden2016holographic,Walter:2020zvt} to explicitly construct the corresponding quantum states and find a new characterization from the perspective of quantum information theory. We leave these questions for future investigation.

\vspace{3mm}

\noindent {\it Acknowledgments:} V.H and M.R. would like to thank the Centro de Ciencias de Benasque Pedro Pascual for hospitality during the program ``Gravity - New perspectives from strings and higher dimensions'', where this work was carried out. T.H. is supported by the Heising-Simons Foundation ``Observational Signatures of Quantum Gravity'' collaboration grant 2021-2817, the U.S. Department of Energy grant DE-SC0011632, and the Walter Burke Institute for Theoretical Physics. V.H. has been supported in part by the U.S. Department of Energy grant DE-SC0009999 and by funds from the University of California. M.R is supported by funds from the University of California.

\bibliography{apssamp}

\end{document}